\documentclass[11pt]{article}

\usepackage{epsfig, amssymb, amsmath}
\usepackage{latexsym, graphics, graphicx}
\usepackage{proof}
\usepackage{algorithmic}
\usepackage{algorithm}
\usepackage{amsthm} 
\usepackage{amsxtra}
\usepackage{float}
\usepackage{authblk}
\usepackage{mathtools}
\usepackage{mathptmx} 
\usepackage{doc}
\usepackage{enumerate}
\usepackage{changepage}
\usepackage{tikz}
\usetikzlibrary{positioning}
\tikzset{>=stealth}

\DeclareMathAlphabet{\mathcal}{OMS}{cmsy}{m}{n}

\newcommand{\ignore}[1]{}

   \newenvironment{my-enumerate}[1][(1)]{
\vspace{1ex}
\begin{adjustwidth}{1em}{}
    \begin{enumerate}[#1]
}{
    \end{enumerate}
\end{adjustwidth}
\vspace{1ex}
}

\newcommand{\X}{\mathcal{X}}

\newcommand{\eqq}{=^?}

\DeclareMathOperator{\pos}{\mathcal{P}\mathit{os}}
\DeclareMathOperator{\fpos}{\mathcal{F\!P}\mathit{os}}
\DeclareMathOperator{\var}{\mathcal{V}\mathit{ar}}

\DeclareMathOperator{\RHS}{\mathit{RHS}}
\DeclareMathOperator{\Mgu}{\mathit{mgu}}

\newtheorem{defn}{Definition}

\ignore{

}

\newcommand{\mylongrightarrow}{\xrightarrow{\hspace*{0.75cm}}}

\theoremstyle{plain}
\newtheorem{theorem}{Theorem}[section]
\newtheorem{Lemma}[theorem]{Lemma}
\newtheorem{corollary}[theorem]{Corollary}
\newtheorem{claim}[theorem]{Claim}

\theoremstyle{definition}

\newtheorem{proposition}[theorem]{Proposition}

\ignore{
\title{Notes on Lynch-Morawska Systems}
\titlerunning{Notes on Lynch-Morawska Systems}

\author{
    Daniel S. Hono II\inst{1}
    \and
    Namrata Galatage\inst{1}
    \and
    Kimberly A. Gero\inst{2}
    \and
    Paliath Narendran\inst{1}
    \and 
    Ananya Subburathinam\inst{1}
    
}
\authorrunning{D. S. Hono II, N. Galatage, K.A. Gero, P. Narendran, and A. Subburathinam}
\institute{
    University at Albany---SUNY (USA)  \\
    \email{\{dhono,ngalatage,dran,asubburathinam\}@albany.edu}\\
    \and 
    The College of Saint Rose (USA) \\
    \email{gerok@strose.edu}
    
}

\date{}
}

\begin{document}

\begin{titlepage}

{\vspace*{-1in}\hspace*{-.5in}
\parbox{7.25in}{
\setlength{\baselineskip}{13pt}
\makebox{\ }\hfill {\footnotesize College of Engineering and Applied Sciences} \\
\makebox{\ }\hfill {\footnotesize Computer Science Department} \\
\makebox{\ }\\
}

\vspace{-.775in}}

\epsfxsize=3.15in
\epsfclipon
\hspace{-0.5in}{\raggedright{
\epsffile{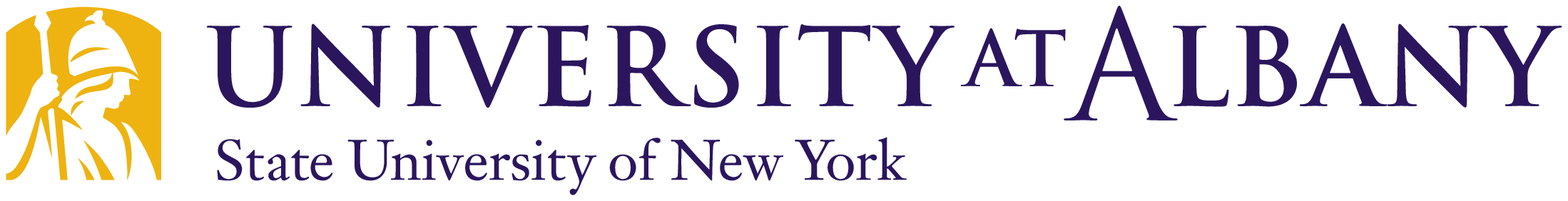}
}}

\vspace{2in}

\begin{center}
{\huge\bf Notes on Lynch-Morawska Systems}\\[+25pt]
\end{center}

\vspace{1.5in}

\begin{center}
{\large\bf 
Daniel S. Hono II\\[+3pt]
Namrata Galatage\\[+3pt]
Kimberly A. Gero\\[+3pt]
Paliath Narendran\\[+3pt]
Ananya Subburathinam}\\

\end{center}
\end{titlepage}

\date{}

\thispagestyle{plain}

\begin{abstract}
In this paper we investigate convergent term rewriting systems that
conform to the criteria set out by Christopher Lynch and Barbara
Morawska in their seminal paper \emph{``Basic Syntactic Mutation.''}  
The equational
unification problem modulo such a rewrite system is solvable in
polynomial-time. In this paper, we derive properties of such a system
which we call an $LM$-system.  We show, in particular, that the
rewrite rules in an $LM$-system have no left-\ or right-overlaps.

We also show that despite the restricted nature of an $LM$-system,
there are important undecidable problems, such as the \emph{deduction
  problem} in cryptographic protocol analysis (also called the
\emph{cap problem}) that remain undecidable for $LM$-systems.

\medskip{}

\end{abstract}

\section{Introduction}
Unification modulo an equational theory $E$ (equational unification
or $E$-unification) is an undecidable problem
in general. Even in cases where it is decidable, it is often
of high complexity.
In their seminal paper ``Basic Syntactic Mutation'' Christopher Lynch and
Barbara Morawska present syntactic criteria on
equational axioms~$E$
that guarantee a polynomial time algorithm for the 
corresponding $E$-unification
problem. As far as we know these are the only purely syntactic
criteria that ensure a polynomial-time algorithm for unifiability.

\ignore{
We also exhibit
theories whose
unification problems do not fit all of the constraints but still
have polynomial time algorithms. 
}

In~\cite{NotesOnBSM} it was shown that relaxing any of the constraints imposed upon a term-rewriting system 
$R$ by the conditions given by Lynch and Morawska results in an unification problem that is $NP$-Hard. 
Thus, these conditions are tight in the sense that relaxing any of them leads to an intractable unification problem (assuming $P \not = NP$). In this work, we continue to investigate the consequences of the syntactic criteria given in ``Basic Syntactic Mutation". As in ~\cite{NotesOnBSM} we consider the case where $E$ forms a convergent and forward closed term-rewriting system, which we call $LM$-Systems. This definition differs from that of~\cite{LynchMorawska} in which $E$ is saturated by paramodulation and not necessarily convergent. The criteria of \emph{determinism} introduced in ~\cite{LynchMorawska} remains essentially unchanged.

We give a structural characterization of these systems by showing that
if $R$ is an $LM$-System, then there are no overlaps between the
left-hand sides of any rules in $R$ and there are no forward-overlaps
between any right-hand side and a left-hand side. This
characterization shows that $LM$-Systems form a very restricted
subclass of term-rewriting systems. Any term-rewriting system that
contains overlaps of these kinds cannot be an $LM$-System. Using these
results, we show that saturation by paramodulation is equivalent to
forward-closure when considering convergent term-rewriting systems
that satisfy all of the remaining conditions for $LM$-Systems.

Despite their restrictive character, we show in Section~$5$ that the
cap problem, which is undecidable in general, remains undecidable when
restricted to $LM$-Systems. The cap problem (also called the deduction
problem) originates from the field of cryptographic protocol
analysis. This result shows that $LM$-Systems are yet strong enough to
encode important undecidable problems. The reduction considered is
essentially the same as that given in~\cite{NotesOnBSM} to show that
determining if a term-rewriting system is subterm-collapsing given all
the other conditions of $LM$-Systems is undecidable.

\ignore{
In the interest of brevity we have omitted most of the proofs of results appearing in this paper. The interested reader can
find all of the details in the technical report.\footnote{arxiv link here}
}

\section{Notation and Preliminaries}
We assume the reader is familiar with the usual notions and concepts in term
rewriting systems~\cite{Term} and equational unification~\cite{BaaderSnyd-01}.
We consider rewrite systems over ranked signatures, usually denoted $\Sigma$,
and a possibly infinite set of variables, usually denoted $\X$. The set of all
terms over~$\Sigma$ and~$\X$ is denoted as $T(\Sigma, \X)$. An \emph{equation}
is an ordered pair of terms $(s, \, t)$, usually written as $s \approx t$. Here
$s$~is the left-hand side and $t$~is the right-hand side of the
equation~\cite{Term}. A rewrite rule is an equation $s \approx t$ where $\var(t)
\subseteq \var(s)$, usually written as~$s \to t$. A term rewriting system is a
set of rewrite rules.

A set of equations $E$ is
\emph{subterm-collapsing}\footnote{Non-subterm-collapsing theories are called
\emph{simple} theories in~\cite{BHSS}} if and only if there are terms $t$ and
$u$ such that $t$ is a proper subterm of $u$ and $E$ $\vdash$ $t \approx u$ (or
$t =_E^{} u$)~\cite{BHSS}.
A set of equations $E$ is
\emph{variable-preserving}\footnote{Variable-preserving theories are also called
\emph{non-erasing} or \emph{regular} theories~\cite{Term}.}
if and only if for every equation $t \approx u$ in $E$,
$\var(t) = \var(u)$~\cite{Ohlebusch95}.
A term rewriting system is \emph{convergent} if and only
if it is confluent and terminating~\cite{Term}. If $R$ is a term rewriting system denote by $IRR(R)$ the set of all $R$-irreducible terms, i.e. the set of all $R$ normal forms. 

A term $t$ is \emph{$\bar{\epsilon}$-irreducible} modulo a rewriting system~$R$ if and only if
every \emph{proper} subterm of~$t$ is irreducible.
A term $t$ is an {\em innermost redex\/} of a rewrite system $R$ if all
proper subterms of $t$ are irreducible and $t$ is \emph{reducible},
i.e., $t$ is \emph{$\bar{\epsilon}$-irreducible} and $t$ is
an instance of the
left-hand-side of a rule in $R$. 

The following definition is used in later sections to simplify the 
exposition. It is related to the above notion of $\bar{\epsilon}$-irreducible. 
Specifically, we define a normal form of a term $t$ that depends
also on the reduction path used. More formally, we have: 

\begin{defn} Let $R$ be a rewrite-system. A term~$t$ is said to be
an \textbf{$\bar{\epsilon}$-normal-form} of a term~$s$ if and only if
$s \, \rightarrow_{R}^* \, t$, no proper subterm of~$t$ is reducible
and none of the rewrites in the sequence of reduction steps is at the
root.
\end{defn}

Given a set of equations $E$, the set of ground instances of $E$ is denoted by
$Gr(E)$. We assume a reduction order $\prec$ on $E$ which is total on ground
terms.  We extend this order to equations as $(s \approx t) \prec (u \approx
v)$ iff $\{s, t\} \prec_{mul} \{u, v\}$, where $\prec_{mul}$ is the multiset
order induced by $\prec$.  An equation~$e$ is \emph{redundant in $E$} if and
only if every ground instance $\sigma(e)$~of~$e$ is a consequence of equations
in~$Gr(E)$ which are smaller than $\sigma(e)$ modulo~$\prec$~\cite{LynchMorawska}.

\subsection{Paramodulation}
Lynch and Morawska define \emph{paramodulation}, which is an extension to the critical pair rule. 
Since our focus is only on convergent term rewriting
systems, this definition can be modified for rewrite rules as the following inference rule:
\[\infer{\sigma(u[t]_p) \approx \sigma(v)}{u[s']_p \approx v &\quad s \to t}\]
where $\sigma = \Mgu(s \eqq s')$ and $p \in \fpos(u)$. A set of equations
$E$ is \emph{saturated by paramodulation} if all inferences among equations in $E$
using the above rule are redundant.

In this work, we consider rewrite systems that are \emph{forward-closed} as
opposed to saturated by paramodulation. A later result of section 4 will
show that saturation by paramodulation and forward-closure are equivalent
properties when all of the other $LM$-conditions obtain.

\section{Forward Closures}

Following Hermann~\cite{Hermann}
the {\em forward-closure\/} of a convergent term rewriting system~$R$ is
defined in terms of the following operation on rules in~$R$: let
$\rho_1^{} : l_1^{} \rightarrow r_1^{}$ and
$\rho_2^{} : l_2^{} \rightarrow r_2^{}$
be two rules in~$R$ and let $p \in {\fpos}(r_1^{})$. Then
\[ \rho_1^{} \, \rightsquigarrow_p^{} \, \rho_2^{} ~ = ~
\sigma (l_1^{} \rightarrow r_1^{}[r_2^{}]_p^{}) \] where
$\sigma = mgu( r_1^{}|_p^{} =_{}^? l_2^{} )$.
Given rewrite systems $R_1^{}$ and $R_2^{}$ such that $R_2^{} \subseteq R_1^{}$,
we define $R_1^{} \rightsquigarrow R_2^{}$ as the rules in: 

\begin{center}
$\left\{ (l_1^{} \rightarrow r_1^{}) \, \rightsquigarrow_p^{} \, (l_2^{} \rightarrow r_2^{}) ~ 
\big| ~ (l_1^{} \rightarrow r_1^{}) \in R_1^{}, \;
(l_2^{} \rightarrow r_2^{}) \in R_2^{} \; \mathrm{and} \;
p \in {\fpos}(r_1^{}) \right\}$ 
\end{center}

which are not redundant in~$R_1^{}$.

We now define
\begin{align*}
    {FC}_0^{} (R) &= R
    \shortintertext{and}
    {FC}_{k+1}^{} (R) &= {FC}_k^{}(R) \; \cup \; ({FC}_k^{}(R) \rightsquigarrow R)
    \shortintertext{for all $k \ge 0$. Finally,}
    FC(R) &= \bigcup\limits_{i=1}^{\infty} \, FC_i^{} (R)
    \intertext{Note that $FC_j^{} (R) \subseteq FC_{j+1}^{} (R)$ for all $j \ge 0$. A set of rewrite rules $R$ is forward-closed if and only if $FC(R) = R$. We
    also define the sets of ``new rules''}
    NR_0(R) &= R
    \shortintertext{and}
    NR_{k+1}(R) &= FC_{k+1}(R) \smallsetminus FC_k(R)
    \intertext{Then}
    FC(R) &= \bigcup\limits_{i=1}^{\infty} \, NR_i^{} (R)
\end{align*}

The following theorem, shown in~\cite{BGLN}, gives
necessary and sufficient conditions for a rewrite-system
to be forward-closed. This property will be used repeatedly
in the sequel below. 

\begin{theorem} \emph{\cite{BGLN}}
\label{thm-fc-irb}
A convergent rewrite system $R$ is forward-closed if and only if every
innermost redex can be reduced to its $R$-normal~form in one step.
\end{theorem}

We next show that there are ways to reduce 
a rewrite system $R$ while still maintaining the properties
that we are interested in. More precisely, given a convergent,
forward-closed rewrite system $R$ we can reduce the right-hand
sides of rules in $R$ while maintaining convergence, forward-closure
and the equational theory generated by $R$. 

Let $R$ be a convergent rewrite system. Following~\cite{gramlich}
we define \[ R\!\downarrow ~ = ~ \left\{ l \rightarrow r\!\downarrow
~ | ~ (l \rightarrow r) \in R \vphantom{b_b^b} \right\} \]

\begin{Lemma}
\label{RightReduceEquiv}
Let $R$ be a convergent, forward-closed rewrite system. Then $R\!\downarrow$ is
convergent, equivalent to~$R$ (i.e., they generate the same 
congruence), and forward-closed.
\end{Lemma}

\begin{proof}
The cases of convergence and equivalence of $R\!\downarrow$ are
handled in~\cite{gramlich}; it remains to show that
$R\!\downarrow$ is forward-closed. By Theorem~\ref{thm-fc-irb} it
suffices to show that every innermost redex
modulo~$R\!\downarrow$ is reducible to its normal form in a
single step.

Let $t$ be an innermost redex of $R\!\downarrow$. Since the
passage from $R$ to $R\!\downarrow$ preserved the left-hand sides
of the rules, $t$ must also be an innermost redex of~$R$. Thus,
$\exists l \rightarrow r \in R$ such that $\sigma(r) \in
IRR(R)$. Then $l \rightarrow r\!\downarrow \in
R\!\downarrow$. Since $r \rightarrow^{!} r\!\downarrow$ and
$\rightarrow$ is closed under substitutions, we have that
$\sigma(r) \rightarrow^{!} \sigma(r\!\downarrow)$, but
$\sigma(r)$ is irreducible, thus $\sigma(r) =
\sigma(r\!\downarrow)$. Therefore, $t$ reduces to its normal form
modulo $R\!\downarrow$ in a single step.
\end{proof}

A convergent rewrite system~$R$ is \emph{right-reduced} if and only if
$R ~ = ~ R\!\downarrow$.

From the above lemma, it is clear that right-reduction does not affect
forward-closure.  However, full interreduction, where one also deletes
rules whose left-hand sides are reducible by \emph{other} rules, will
not preserve forward-closure. The following example illustrates this:
\begin{eqnarray*}
f(x, i(x)) & \rightarrow & g(x)\\
g(b) & \rightarrow & c\\
f(b, i(b)) & \rightarrow & c
\end{eqnarray*}

The last rule can be deleted since its left-hand side is reducible by the
first rule. This will preserve convergence, but forward-closure will be lost since
$f(b, i(b))$, an innermost redex, cannot be reduced in \emph{one step} to~$c$
in the absence of the third rule.

However, the following lemma enables us to do a restricted deletion of superfluous rules:
\begin{Lemma}
\label{semileftreduced}
 Let $R$ be a convergent, forward-closed, term rewriting system. Let 
 $l_i \rightarrow r_i \in R$ for $i \in \{1, 2\}$ such that 
 $\exists p \in \fpos(l_1) : p \not = \epsilon \text{ and } l_{1}|_{p} = \sigma(l_2)$
 for some substitution $\sigma$. That is, $l_1$ contains a proper subterm that is an instance
 of the left-hand side of another rule in $R$. Then, 
$R' = R \smallsetminus \left\{ l_1 \rightarrow r_1 \right\}$
 is convergent, forward-closed and equivalent to~$R$. 
\end{Lemma}

\begin{proof}
For the sake of deriving a contradiction, assume that
$R'$ is not forward-closed. Then there must exist
some term $t$ such that $t$ is an innermost redex
modulo $R'$ and $t$ is not reducible to its normal
form in a single step. 

Since $R$ is forward-closed $t$ must be reducible to
its normal form in a single step modulo $R$, and since 
$R$ differs from $R'$ by the rule $l_1 \rightarrow r_1$, 
then it must be that $t = \theta(l_1)$ for some substitution
$\theta$.   

However, since $l_{1}|_p = \sigma(l_2)$ we have that $\theta(l_{1})|_p
= \theta(\sigma(l_2))$, but since $\theta(l_{1}) = t$ we also have
that $\theta(l_1)|_{p} = t|_{p}$. Thus, $t|_p = \theta(\sigma(l_2))$,
but this contradictions the assumption that $t$ is an innermost redex
as $p \not = \epsilon$. Therefore, $R'$ must be forward-closed. 

The termination of $R'$ follows from the fact that $R$ is terminating
and $R' \subset R$. Towards showing confluence and equivalence we show that $IRR(R') = IRR(R)$. First, it is clear that $IRR(R) \subseteq IRR(R')$ 
since $R' \subseteq R$. For the reverse containment, suppose that $t \in IRR(R')$ 
but $t \not \in IRR(R)$. That is, $t$ must be reducible modulo $R$. However, since $R$ 
and $R'$ differ by only a single rule, it must be the case that $t$ is reducible by
the rule $l_1 \rightarrow r_1$, but $l_1$ is reducible by $l_2 \rightarrow r_2$, thus
$t$ would also be reducible by $l_2 \rightarrow r_2 \in R'$. But this contradicts the assumption that $t \in IRR(R')$. 

Now, suppose that $R'$ is not confluent. Then there must a term $t$
with two distinct $R'$-normal forms $t'$ and $t''$, but $t' \downarrow_{R}
t''$ as $R$ is confluent. However, this contradicts the above fact
that $IRR(R) = IRR(R')$ as at least one of $t', t''$ must be
reducible. Thus, $R'$ is confluent, and thus we have established that
$R'$ is convergent.

Finally, we show that $R'$ is equivalent to $R$. Since $R' \subset R$
we have that $\leftrightarrow_{R'}^* \, \subseteq \,
\leftrightarrow_{R}^*$. For the reverse containment, suppose there are
two distinct $R'$-irreducible terms $s$ and $t$ such that $s
\leftrightarrow_{R}^* t$, but then since $IRR(R) = IRR(R')$, $s = t$,
which contradicts the assumption that they are distinct terms.
Thus $R'$ is equivalent to $R$.  
\end{proof}

We call systems that have no rules such that the conditions of
Lemma~\ref{semileftreduced} obtain \textit{almost-left-reduced}.
Note, however, that they are not fully left-reduced as there is still
the possibility of overlaps at the \emph{root.} If a convergent,
forward-closed and right-reduced $R$ is not almost-left-reduced, then
the above lemma tells us that we may delete such rules and obtain an
equivalent system.

\section{Lynch-Morawska Conditions}

In this section we define the Lynch-Morawska conditions. We also derive some
preliminary results on convergent term rewriting systems that satisfy
the Lynch-Morawska conditions. 

A new concept introduced by Lynch and Morawska is that of
a {\em Right-Hand-Side Critical Pair,\/} defined
as follows:

\vspace{2 mm} 

{\large
\begin{center} 
\fbox{
\begin{minipage}{3.75in}
\centerline{$\infer{s\sigma \approx u\sigma}
        {s \approx t ~ \qquad \qquad ~ u \approx v}$}

\vspace{0.1in}
{$\mathrm{where ~} s\sigma \nprec t\sigma, u\sigma \nprec v\sigma, \sigma 
= mgu(v, t) ~ \mathrm{and} ~ s\sigma \neq u\sigma$}
\end{minipage}
}
\end{center}
}

Since our focus is only on convergent term rewriting systems, this
definition can be modified as follows:

\vspace{2 mm} 

{\large
\begin{center} 
\fbox{
\begin{minipage}{3.75in}
\centerline{$\infer{s\sigma \approx u\sigma}
        {s \rightarrow t ~ \qquad \qquad ~ u \rightarrow v}$}

\vspace{0.1in}
{$\mathrm{where ~} \sigma = mgu(v, t) ~ \mathrm{and} ~ s\sigma \neq u\sigma$}
\end{minipage}
}
\end{center}
}

\vspace{2 mm} 

For instance the right-hand-side critical pair
$f(x, s(y)) \rightarrow s(f(x, y))$ and
$f(s(x), y) \rightarrow s(f(y, x))$ is
$f(x, s(x)) \approx f(s(x), x)$. Also note (as pointed out in~\cite{LynchMorawska})
that the rule $f(x, x) \rightarrow 0$ has a right-hand-side critical pair:
$f(x, x) \approx f(x', x')$.

For an equational theory $E$, {\em RHS(E)\/} = $\left\{ \, e ~ \big| ~ e \right.$
is the conclusion of a Right-Hand-Side Critical Pair inference of two
members of $E$ $\left. \vphantom{\big|} \right\}$ $\cup$ $E$~\cite{LynchMorawska}.

A set of equations~$E$ is \emph{quasi-deterministic} if and only if
\begin{enumerate}
    \item No equation in~$E$ has a variable as its left-hand side or right-hand
        side,

    \item No equation in $E$ is \emph{root-stable}---i.e., no equation has the
        same root symbol on its left- and right-hand side, and

    \item $E$ has no \emph{root pair repetitions}---i.e., no two equations in
        $E$ have the same pair of root symbols on their sides.
\end{enumerate}

The following lemma was proved in~\cite{NotesOnBSM}.

\begin{Lemma}
    Suppose $R$ is a variable-preserving convergent rewrite system and $R$ is
    quasi-deterministic.  Then $\RHS(R)$ is not quasi-deterministic if and only
    if $\RHS(R)$ has a root pair repetition.
    \label{lemma-quasi}
\end{Lemma}

A theory $E$ is \emph{deterministic} if and only if it is quasi-deterministic
and non-subterm-collapsing.

\medskip{}

\noindent
A \emph{Lynch-Morawska term rewriting system} or \emph{LM-system} is a
convergent, almost-left-reduced and right-reduced term rewriting
system~$R$ which satisfies the following conditions:
\begin{itemize}
\item[(i)] $R$ is non-subterm-collapsing, \\[-18pt]
\item[(ii)]
$R$ is
forward-closed, and \\[-18pt]
\item[(iii)] $RHS(R)$ is quasi-deterministic. 
\end{itemize}

The goal of the remainder of this section is to show that, given an \emph{LM-system} $R$, there 
can be no overlaps between the left-hand sides of any rules in $R$ and 
that there can be no forward-overlaps. These notions are defined precisely
below. Further, we use those results to derive the equivalence of forward-closure
and saturation by paramodulation when $R$ is an \emph{LM-system}. 

These results show that \emph{LM-systems} are a highly restrictive subclass 
of term-rewriting systems. However, in a later section, we show that
there are important decision problems that remain undecidable when
restricted to \emph{LM-systems}. 

The first of these results, Lemma~\ref{ExactlyTwo} and its proof, are used
multiple times to prove other results. It concerns how two terms $s$ and $t$, 
such that $s$ is an innermost redex and $t$ is $\bar{\epsilon}$-irreducible, can 
be joined. It establishes that there are only two possible cases, and further, only
one of these cases can hold at a time. 

\begin{Lemma}
\label{ExactlyTwo}
Let $R$ be an $LM$-system, $s = f(s_1^{}, \ldots , s_m^{})$ an innermost redex
and
$t = g(t_1^{}, \ldots , t_n^{})$ an $\bar{\epsilon}$-irreducible term such that~$f \neq g$.
Then $s$ and $t$ are joinable modulo~$R$ if and only if \emph{exactly one} of the following 
conditions holds:
\begin{itemize}
\item[(a)] there is a unique rule $l \rightarrow r$ with root pair $(f, g)$ and
\( \displaystyle{s \mathop{\mylongrightarrow}_{l \rightarrow r}^{} t} \), $\; ~ ~$\emph{or}

\item[(b)] there are unique rules $l_1^{} \rightarrow r_1^{}$ and
  $l_2^{} \rightarrow r_2^{}$ with root pairs $(f, h)$ and $(g, h)$
  such that $\displaystyle{s
    \mathop{\xrightarrow{\hspace*{0.75cm}}}_{l_1 \rightarrow r_1}^{}
    \widehat{t}}$ and $\displaystyle{t
    \mathop{\xrightarrow{\hspace*{0.75cm}}}_{l_2 \rightarrow r_2}^{}
    \widehat{t}}$ for some term~$\widehat{t}$.
\end{itemize}

\end{Lemma}

\begin{proof}

$(\Rightarrow)$ There are two cases to consider.
 \begin{itemize}
  \item[\emph{(i)}] $t = g(t_1, \ldots t_n)$ is in normal form. 
  \item[\emph{(ii)}] $t$ is an innermost redex. 
 \end{itemize}

\textit{Case (i)}. Since $s \downarrow_R t$ by assumption, $s$ must
reduce to $t$ in a single step (as $s$ is an innermost redex and
$t$ cannot be reduced modulo $R$ any further). Therefore
$\exists! \; \rho = l \rightarrow r \in R$ that will reduce $s$
to $t$, and this rule is unique as there can be no root-pair
repetitions.

\textit{Case (ii)}. Since $t$ is an innermost redex and $s$ is
$\bar{\epsilon}$-irreducible (and therefore also an innermost redex in this
case) they both must be reducible to their normal forms in one
step and since $s \downarrow_R t$ these normal forms must be
equal. Let $\widehat{t}$ be the normal form of $s$ and $t$. Since
there can be no root-stable equations $\widehat{t}(\epsilon) = h$ and $h
\not= f, g$. Thus, $\exists! \; \rho_{i} = l_i \rightarrow r_i
\in R$ for $i \in \{1, 2\}$ with root pairs $(f, h)$ and $(g, h)$
respectively such that $s \rightarrow_{\rho_{1}} \widehat{t}$ and 
$t \rightarrow_{\rho_{2}} \widehat{t}$.

It remains to show that both case~$(a)$ and case~$(b)$ cannot obtain
simultaneously. For the sake of deriving a contradiction, assume
that both case~$(a)$ and case~$(b)$ hold. The only way this could
occur is with case~$(ii)$ above. Without loss of generality assume
that $\mathcal{V}ar(\rho_1) \cap \mathcal{V}ar(\rho_2) =
\varnothing$. Since $s$ and $t$ are $\bar{\epsilon}$-irreducible we have that
$\sigma_1(r_1) = \widehat{t} = \sigma_2(r_2)$ where $\sigma_1 =
mgu(l_1 \, \lesssim_{}^? \, s)$ and $\sigma_2 = mgu(l_2 \,
\lesssim_{}^? \, t)$. And so by defining $\sigma := \sigma_1 \cup
\sigma_2$ we have that $\sigma(r_1) = \sigma(r_2)$, i.e.,
$\sigma$ is a unifier of $r_1$ and~$r_2$.

We can thus perform a RHS inference step using $\rho_1$ and $\rho_2$, i.e.,

\[\frac{ l_1 \rightarrow r_1 \; \qquad \; l_2 \rightarrow r_2}{\theta(l_1) \approx \theta(l_2)} \] 

\noindent
to get that $\theta(l_1) \approx \theta(l_2) \in RHS(R)$ where
$\theta = mgu(r_1 =^{?} r_2)$. This equation has root-pair $\{f,
g\}$. 
\ignore{
Now consider the rewrite rule given by $\theta(l_1)
\rightarrow \theta(l_2)$. Recall that $\sigma \in
\mathcal{U}(r_1, r_2)$, but $\theta$ is the mgu of $r_1$ and
$r_2$ and therefore we have that $\theta \precsim \sigma$ and
therefore we get $(\exists \delta \in \textbf{Sub})[ \sigma =
  \delta \theta]$. From before we know that $\sigma(l_1) = s$ and
thus $\rho\theta(l_1) = s$ which gives us that $\theta(l_1)$
matches with $s$. We can hence rewrite $s$ to $\rho\theta(l_2) =
t$, but by assumption $l \rightarrow r$ is unique, therefore
$\theta(l_1) \approx \theta(l_2) = l \approx r$, but then $r$ is
reducible by $l_2 \rightarrow r_2$ which contradicts the
assumption that $R$ is right-reduced.
}
But since $RHS (R)$ can have no root-pair repetition, it must be
that $\theta(l_1) \approx \theta(l_2) = l \approx r$, but then $r$ is
reducible by $l_2 \rightarrow r_2$ which contradicts the
assumption that $R$ is right-reduced.

$(\Leftarrow)$ If exactly one of the two conditions hold, then
$s$ and $t$ are joinable by definition.
\end{proof}

We present an interesting consequence of the above lemma. Namely,
if a term $t$ such that $t(\epsilon) = f$ has as normal form a
term $s$ such that $s(\epsilon) = g$ for $g$ differing from $f$,
then we can establish some information on the rules in $R$. 

\begin{corollary}
\label{NormalFormPair}
If $f( s_1^{} , \ldots , s_m^{} ) \; \rightarrow_R^! \; g( t_1^{} ,
\ldots , t_n^{} )$ where~$f \neq g$, then $(f, \, g)$ is a root pair
in~$R$.
\end{corollary}

\begin{proof}
 Let $s = f(s_1, \ldots, s_n)$ and $t = g(t_1, \ldots, t_m)$. Suppose
 that $s \rightarrow_{R}^{!} t$. Then by definition $s \downarrow_R
 t$. Since $t$ is a normal form it is also an $\bar{\epsilon}$-irreducible term
 modulo~$R$.
 
 Let $s' = f(s_1', \ldots, s_n')$ be the term obtained from $s$ by
 reducing all top-level subterms of $s$ to their normal forms modulo~
 $R$.  Thus we have the following situation: $s \rightarrow_{R}^{*} s'
 \rightarrow_{R}^{!} t$ as $s'$ must be an $\bar{\epsilon}$-irreducible term and
 hence also an innermost redex in this case. Therefore we can apply
 \textit{case~(a)} of Lemma~\ref{ExactlyTwo} to conclude that there must be a
 unique rule $\rho = l \rightarrow r \, \in \, R$ that reduces $s'$ to $t$
 with root pair~$(f, g)$.
\end{proof}

Given Lemma~\ref{ExactlyTwo}, we can immediately derive two results
that will be useful in proving the main result of this section.  

\begin{corollary}
\label{UniqueRule}
Suppose $l \rightarrow r \in R$ is a rule with root-pair $(f, g)$, $s
= f(s_1, \ldots, s_m)$ and $t = g(t_1, \ldots, t_n)$, and $s$ and $t$
are $\bar{\epsilon}$-irreducible. Then, $s \downarrow_{R} t$ if and only if
$\;$ \(
\displaystyle{s ~ \mathop{\xrightarrow{\hspace*{0.75cm}}}_{l^{}
    \rightarrow r^{}}^{} ~ t}. \)
\end{corollary}

\begin{proof}
This result follows from Lemma~\ref{ExactlyTwo} and its proof. 
\end{proof}

\begin{corollary}
\label{UniqueRule2}
Let $R$ be an LM-System. Suppose $l \rightarrow r \in R$ is a rule
with root-pair $(f, g)$ and $s = f(s_1, \ldots, s_m)$ and $t = g(t_1,
\ldots, t_n)$ be terms that are joinable. Let $\widehat{s_1}, \ldots ,
\widehat{s_m}, \, \widehat{t_1}, \ldots, \widehat{t_n}$ be
respectively the normal forms of 
$s_1 , \ldots , s_m, \, t_1 , \ldots
, t_n$.  Then \[ f(\widehat{s_1}, \ldots , \widehat{s_m}) ~
\mathop{\xrightarrow{\hspace*{0.75cm}}}_{l^{} \rightarrow r^{}}^{} ~
g(\widehat{t_1}, \ldots , \widehat{t_n}) . \]
(Thus the normal form of $s$ and $t$ is an instance of the right-hand side~$r$.)
\end{corollary}

\begin{proof}
Since $s$ and $t$ are joinable, it must be the case that
$\widehat{s} = f(\widehat{s_1}, \ldots, \widehat{s_m})$ and
$\widehat{t} = g(\widehat{t_1}, \ldots, \widehat{t_n})$ are
joinable, but since each $\widehat{s_i}$ and $\widehat{t_i}$ are
in normal form, $\widehat{s}$ and $\widehat{t}$ must be
$\bar{\epsilon}$-irreducible, therefore we can apply Corollary~\ref{UniqueRule}.
\end{proof}

The above two corollaries, along with Lemma~\ref{ExactlyTwo}, allows
us to state the first of the results concerning the
non-overlapping property of LM-systems. The following establishes that
there can be no overlaps between left-hand sides of two rules occurring at
the root-position and that there can be no overlaps between a right-hand side
of a rule and a left-hand side of another rule at the root position. This is achieved
by showing that these terms cannot be unified. 

\begin{corollary}
\label{noRootOverlaps}
Let $R$ be an LM-System and let
$l_1^{} \rightarrow r_1^{}$ and $l_2^{} \rightarrow r_2^{}$ be 
\emph{distinct} rules in~$R$. Then
\begin{itemize}
\item[(a)] $l_1^{}$ and $l_2^{}$ are not unifiable, and

\item[(b)] $r_1^{}$ and $l_2^{}$ are not unifiable.
\end{itemize}
\end{corollary}

\begin{proof}
Suppose that the rule $l_1 \rightarrow r_1$ has root pair $(f, g)$ and
the rule $l_2 \rightarrow r_2$ has root pair $(h, i)$. We show that
each case above leads to a contradiction.

Thus, for case (a), towards deriving such a contradiction suppose that
$\theta ~ = ~ mgu(l_1^{} \; =_{}^? \; l_2^{})$. Then, we have that
$\theta(l_1) \rightarrow \theta(r_1)$, and so $\theta(l_1)$ and
$\theta(r_1)$ are obviously joinable. Note that, since $l_1$ is
unifiable with $l_2$, $l_1(\epsilon) = l_2(\epsilon)$. Thus, $f = h$
above, and $g \not = i$ as the contrary would induce a root-pair
repetition in $R$.

By applying Corollary~\ref{UniqueRule2} to $\theta(l_1)$ and
$\theta(r_1)$ we can conclude that their normal form must be some
instance of $r_1$. Likewise, we can apply the same corollary on
$\theta(l_1)$ and $\theta(r_2)$, i.e. their normal form must be an
instance of $r_2$.

But since $R$ is convergent, $\theta(r_1)$ and $\theta(r_2)$ 
must be joinable and thus must have the same
normal form, but this is impossible as $r_1$ and $r_2$ have different
root symbols.

For case (b), suppose that $\beta ~ = ~ mgu(r_1^{} \; =_{}^? \; l_2^{})$. 
Thus $\beta (l_1^{}) \, \rightarrow \, \beta (r_1^{}) \, \rightarrow \,
\beta(r_2^{})$. Hence $\beta (l_1^{})$ and $\beta (r_1^{})$ are
joinable, and $\beta (l_1^{})$ and $\beta (r_2^{})$ are joinable. The
rest of the argument is the same as for the above case.
\end{proof}

Next, we work towards showing that the other possible 
overlaps cannot occur either. The next lemma, and its 
extension, are used towards this goal. The technical 
result is used in the proofs of various other lemmas 
and corollaries.

\begin{Lemma}
\label{CommutingSquare}

Let $R$ be an LM-System, and suppose $f(s_1, \ldots, s_m) \displaystyle{ ~ \mathop{\xrightarrow{\hspace*{0.75cm}}}_{l^{}
    \rightarrow r^{}}^{}} ~ g(t_1, \ldots, t_n)$. Then the following diagram commutes. 
    
 \begin{center}
  \begin{tikzpicture}
    \node (A) at (0, 0) {$f(s_1, \ldots, s_m)$};
    \node[right=of A] (B) {$g(t_1, \ldots, t_n)$};
    \node[below=of A] (C) {$f(s_1', \ldots, s_m')$};
    \node[below=of B] (D) {$g(t_1', \ldots, t_n')$};
    
    \draw[->] (A)--(B) node [below, pos=.5] {$l \rightarrow r$};
    \draw[->] (A) --(C) node [left, pos=.8] {$*$}; 
    \draw[->, dashed] (C) -- (D) node [below, pos=.5] {$l \rightarrow r$}; 
    \draw[->] (B) -- (D) node [right, pos=.8] {$*$};
  \end{tikzpicture}
  \end{center}
  
\noindent where $s_1', \ldots, s_m', t_1', \ldots, t_n'$ are the normal forms of $s_1, \ldots, t_n$ respectively. 
\end{Lemma}

\begin{proof}
  We have that $f(s_1, \ldots, s_m) \to^*
 f(s_1', \ldots, s_m')$ and $f(s_1, \ldots, s_m) \to^+ g(t_1', \ldots,
 t_n')$, and since $R$ is confluent $s' = f(s_1', \ldots, s_m') \;
 \downarrow_R \; g(t_1', \ldots, t_n') = t'$. Since both $s'$ and $t'$
 are $\bar{\epsilon}$-irreducible, they must be joinable in a single step. Let
$\widehat{t}$ be their normal form. There are
 three cases corresponding to Lemma~\ref{ExactlyTwo}. Cases $1$ and $2$ below
 correspond to case $(a)$ of Lemma~\ref{ExactlyTwo} and case $3$ corresponds to
 case $(b)$ of the Lemma~\ref{ExactlyTwo}. We have:
 
 \begin{itemize}
  \item[1.] $s' \rightarrow t'$ by a unique rule in $R$ with root-pair $(f, g)$,
  \item[2.] $t' \rightarrow s'$, by a unique rule in $R$ with root-pair $(g, f)$, 
  \item[3.] $s' \rightarrow \widehat{t}$ and $ t' \rightarrow \widehat{t}$, by two unique rules with root-pairs $(f, h)$ and $(g, h)$ for some $h$. 
 \end{itemize}

 However, case 2 leads to a contradiction as it would imply that there
 exists a rule in $R$ with root pair $(g, f)$ which would be a root
 pair repetition in~$E$.
 
 Suppose case 3 were true. Let $\rho_i = l_i \rightarrow r_i \in R$
 for $i \in \{1, 2\}$ be the unique rules with root pairs $(f, h)$ and
 $(g, h)$ respectively that reduce $s'$ and $t'$ to $\widehat{t}$. Without
 loss of generality assume that $\mathcal{V}ar(\rho_1) \cap
 \mathcal{V}ar(\rho_2) = \varnothing$. Then, $r_1$ and $r_2$ are
 unifiable and therefore we can perform a RHS inference step to get
 $\theta(l_1) \approx \theta(l_2)$. However, there is a root-rewrite
 step along the path from $s$ to $t$. Let $l \rightarrow r \in R$ be
 the rule that induces the root-rewrite step. Thus $\theta(l_1)
 \approx \theta(l_2) = l \approx r$ as $l \rightarrow r$ must be
 unique. This implies, however, that $r$ is reducible by $l_2
 \rightarrow r_2$ would contradicts the assumption that $R$ is
 right-reduced.

 We are then only left with case 1. This rule is unique by Lemma~\ref{ExactlyTwo}
 and since $l \rightarrow r$ has root pair $(f, g)$ these rules must
 be the same.
\end{proof}

We now extend the previous result by induction.

\begin{Lemma} 
\label{CommutingLemma}
Let $R$ be an LM-System and suppose $f(s_1, \ldots, s_m)
\rightarrow_{R}^{+} g(t_1, \ldots, t_n)$ are terms such that 
$f \not = g$. Then the following diagram commutes:
  
  \begin{center}
  \begin{tikzpicture}
    \node (A) at (0, 0) {$f(s_1, \ldots, s_m)$};
    \node[right=of A] (B) {$g(t_1, \ldots, t_n)$};
    \node[below=of A] (C) {$f(s_1', \ldots, s_m')$};
    \node[below=of B] (D) {$g(t_1', \ldots, t_n')$};
    
    \draw[->] (A)--(B) node [above, pos=1] {$+$};
    \draw[->] (A) --(C) node [left, pos=.8] {$*$}; 
    \draw[->, dashed] (C) -- (D) node [above, pos=1] {$+$}; 
    \draw[->] (B) -- (D) node [right, pos=.8] {$*$};
  \end{tikzpicture}
  \end{center}
  
\noindent
where $s_i' = s_i\downarrow_R$ and $t_j' = t_j\downarrow_R$ for $1 \leq i \leq n$, $1 \leq j \leq m$.  
\end{Lemma}

\begin{proof}
The proof proceeds by induction on the number of rewrites occurring at
the root along the chain $s = f(s_1\ldots, s_m)$ to $t =
g(t_1, \ldots, t_n)$. Let $q$ be the number of root rewrite steps
occurring as stated above. \\
 
 \noindent
 \textbf{Base Step:} The base step, where $q = 1$, corresponds to
 Lemma~\ref{CommutingSquare}. \\

\noindent
\textbf{Inductive Step:} Assume that the result holds for $q =
k$. We show that the result is also true for $q = k+1$. Since
there can be no root-pair repetitions in~$R$, we must have that
there is a sequence of root pairs starting with $f$ and ending
with $g$ corresponding to the rules used in the reductions. In
the diagram below, the first square commutes by the base case,
and the rest of the chains can be filled in to create commuting
squares by the induction hypothesis up to the $k+1$ root
rewrite. That is:

\begin{center}
 \begin{tikzpicture}
    \node (F) at (0, 0) {$f(s_1, \ldots, s_m)$};
    \node (H0) [right=of F] {$h_{0}(u_1, \ldots, u_s)$};
    \node (DOTS) [right=of H0] {$\cdots$};
    \node (Hk) [right=of DOTS] {$h_{k}(u_1, \ldots, u_j)$};
    \node (G) [right=of Hk] {$g(t_1, \ldots, t_n)$};
    \node (FP) [below=of F] {$f(s_1', \ldots, s_m')$};
    \node (H0P) [below=of H0] {$h_{0}(u_1', \ldots, u_s')$};
    \node (LDOTS) [right=of H0P] {$\ldots$};
    \node (Hkp) [below=of Hk] {$h_{k}(u_1', \ldots, u_j')$};
    \node (Gp) [below=of G] {$g(t_1', \ldots, t_n')$};
    
    \draw[->] (F)--(H0) node [above, midway] {$+$};
    \draw[->] (H0)--(DOTS) node [above, midway] {$*$};
    \draw[->] (DOTS)--(Hk) node [above, midway] {$*$};
    \draw[->] (Hk)--(G) node [above, midway] {$+$};
    \draw[->] (FP)--(H0P) node [above, midway] {$+$};
    \draw[->] (H0P)--(LDOTS) node [above, midway] {$*$};
    \draw[->] (LDOTS)--(Hkp) node [above, midway] {$*$};
    \draw[->] (F)--(FP) node [left, midway] {$*$};
    \draw[->] (H0)--(H0P) node [left, midway] {$*$};
    \draw[->] (Hk)--(Hkp) node [left, midway] {$*$};
    \draw[->] (G)--(Gp) node [left, midway] {$*$};
 \end{tikzpicture}
\end{center}

\noindent
Therefore, we only need to fill in the final square. However, a
similar argument as for the proof of Lemma~\ref{CommutingSquare}
applies to the terms $h_k(u_1, \ldots, u_j) \to^+ g(t_1', \ldots,
t_n')$ and $h_k(u_1, \ldots, u_j) \to^* h_k(u_1', \ldots, u_j')$.
\end{proof}

The following corollary establishes that given an LM-system $R$, then
if $l \rightarrow r \in $ is a any rule with root-pair $(f, g)$ no term 
with root-symbol $f$ can be reduced modulo $R$ to a term with $g$
as its root-symbol.

\begin{corollary}
\label{NoReversals}
Let $f( s_1^{} , \ldots , s_m^{} ) \; \rightarrow \; g( t_1^{} ,
\ldots , t_n^{} )$ be a rule in~$R$ where~$f \neq g$. Then no term
with $g$ as its root symbol can be reduced modulo~$R$ to a term with
$f$ at its root.
\end{corollary}

\begin{proof}
The proof is by contradiction. Suppose that there exists a
reduction chain $g(s_1, \ldots, s_m) \to^{+} f(t_1, \ldots, t_n)$
such that $g \not = f$. Applying Lemma~\ref{CommutingLemma} we obtain a
reduction chain $s' = g(s_1', \ldots, s_m') \to^+ f(t_1',
\ldots, t_n')= t'$ where $s_i' = s_i\downarrow_R$ and $t_j' =
t_j\downarrow_R$ for $1 \leq i \leq m$ and $1 \leq j \leq
n$. Thus, $s'$ and $t'$ are joinable, and $s'$ must be an
innermost redex as it is not in normal form, therefore we can
apply Lemma~\ref{ExactlyTwo} and get two cases.
 
Case ($a$) leads to a contradiction as $(f, g)$ is already a root
pair in $R$ by assumption, thus, $\{f, g\}$ is a root pair in~$E$
and therefore there cannot be a root pair $(g, f)$ as this would
contradict the assumption that $R$ is an $LM$-system.
 
Therefore, since exactly one of the two cases must obtain,
case~($b$) must be true. Thus there exist unique rules
$\rho_i = l_i \to r_i$ for $i \in \{1, 2\}$ such that $s'
\to_{\rho_1} \widehat{t}$ and $t' \to_{\rho_2} \widehat{t}$ with root
pairs $(g, h)$ and $(f, h)$ respectively. However, again we have
that $r_i$ for $i \in \{1, 2\}$ have a common instance and hence
are unifiable, and so we can perform an RHS inference to get
$\theta(l_1) \approx \theta(l_2) \in RHS(R)$ where $\theta =
mgu(r_1 =^? r_2)$. Since we assumed that $l \to r = f(s_1,
\ldots, s_m) \to g(t_1, \ldots, t_m)$ is a rule in $R$ this
inference step must yield that $\theta(l_1) \approx \theta(l_2) =
l \approx r$ as it has root pair $\{f, g\}$. This leads to a
contradiction however, as it implies that $r$ is reducible by
$l_2 \rightarrow r_2$ which contradicts the assumption that $R$
is right-reduced.
\end{proof}

We next establish two corollaries that give us information
about where reductions can take place. 

\begin{corollary}
\label{InnermostReductions} Let $R$ be an LM-System, then
$s = f( s_1^{} , \ldots , s_n^{} ) \; \rightarrow_R^* \; 
f( t_1^{} , \ldots , t_n^{} ) = t \,$
if and only if \[ \forall i \in \{1, \ldots, n\}: [s_i^{} \; \rightarrow_R^* \; t_i^{}] . \]
\end{corollary}

\begin{proof}
The ``if'' part is obvious. For the ``only if'' part,
suppose $s \; \rightarrow_R^* \; t$ as above and 
$s_j \not\rightarrow_R^{*} t_j$ 
for some~$1 \leq j \leq n$. Then some rewrite step in the sequence
from $f(s_1, \ldots, s_n)$ to $f(t_1, \ldots, t_n)$ must have occurred
at the root. Thus there must be a rule with
(directed) root pair $(f, h)$ for some $h \not = f$. However, $t$ has
root symbol~$f$. Thus, a term with $h$ as the root symbol must be
reducible to a term with $f$ at the root symbol, but this contradicts
Corollary~\ref{NoReversals}.
\end{proof}

\begin{defn} Let $s_1^{}, \, s_2^{}$ be terms. A position $p \in \pos(s_1^{}) \cup \pos(s_2^{})$
is said to be an \emph{outermost distinguishing position} between $s_1^{}$
and $s_2^{}$ if and only if $s_1^{} (p) \neq s_2^{} (p)$ and
$s_1^{} (p') = s_2^{} (p')$ for all proper prefixes~$p'$ of~$p$. The set of all
outermost distinguishing positions between two terms $s$ and~$t$
is denoted by $ODP(s, t)$. 
\end{defn}

Note that $ODP(s, s) = \emptyset$.

\begin{corollary}
\label{ODPReductions}
Let $R$ be an LM-System. Then $s \xrightarrow{*} t \; \,
\text{if and only if} \; \,
\forall p \in ODP(s, t): \left[s|_{p} \xrightarrow{*} t|_p\right]$.
\end{corollary}

\begin{proof}
The ``if" direction is straightforward. We prove the ``only if"
direction by induction on triples (w.r.t.\ the induced lexicographic
xorder) $(|s|, |t|, |p|)$ where $s, t$ are terms and $p$ is a position
such that $s \rightarrow_{R}^{*} t$ and $p \in ODP(s, t)$. \\

\noindent
\textbf{Basis.} The ``least" (i.e., lowest in the ordering) 
such triple is $(1, 1, 0)$, corresponding
to terms such as $s = a$ and $t = b$, i.e., constants. Suppose that 
$s \rightarrow^{*} t$. Then, $ODP(s, t) = \{\epsilon\}$ and therefore
$s|_p = s \rightarrow^{*} t = t|_p$, which is exactly the statement of
the corollary. Thus, we can conclude that the base case holds. \\

\noindent
\textbf{Inductive Step.} Assume that the result is true for all
triples $C \prec (|s'|, |t'|, |p'|)$. We show that the result holds
for the triple $(|s'|, |t'|, |p'|)$ itself. Suppose 
$s' \rightarrow^{*} t'$ and $p' \in ODP(s', t')$. Again, since the
result clearly holds for $p' = \epsilon$ we assume that $p' \not =
\epsilon$. Then $s' = f(s_1, \ldots, s_n)$ and 
$t' = f(t_1, \ldots, t_n)$ for some~$f$ and 
terms~$s_1, \ldots , s_n, t_1, \ldots , t_n$.
By Corollary~\ref{InnermostReductions} it must be the case that
$\forall i \in \{1, \ldots, n\} : \left[ \vphantom{b^b} s_i \rightarrow^{*} t_i
\right]$. Since
$|s_i| < |s'|$ and $|t_i| < |t'|$ we invoke the induction hypothesis
to conclude that 
$(\forall i \in \{1, \ldots, n\})(\forall q \in
ODP(s_i, t_i))\left[ \vphantom{b^b} s_i|_q \rightarrow^{*} t_i|_q \right]$. 
But, we then have that
$p' = i \cdot q$ for some $i \in \{1, \ldots, n\}$ and some $q \in
ODP(s_i, t_i)$, and so we have that $s'|_{i \cdot q}^{} = s|_{p'}
\rightarrow^{*} t'|_{i \cdot q}^{} = t'|_{p'}$.

We can therefore conclude that the result must hold for all triples. 
\end{proof}

\begin{defn} (\emph{non-overlay superpositions}) Let $R$ be a rewrite-system, then \[
NOSU\!P(R) ~ := ~ \left\{ \, \sigma( l_1^{} [l_2^{}]_p^{} ) ~ \; \big| \; ~ 
\vphantom{c_c^c} p \in \fpos( l_1^{} ) \smallsetminus \{\epsilon \} ~ \text{and} ~
\sigma = \Mgu( l_1^{} |_p^{} \, =^? \, l_2^{} ) \,, l_1 \rightarrow r_1,\; l_2 \rightarrow r_2 \in R\; \right\} \]
\end{defn}

The previous results of this section are now brought together to prove
the following main results about LM-systems concerning the status 
of overlaps. Namely, we show that there are no non-overlay superpositions
and no forward-overlaps.

We first establish that there are no superpositions occurring 
between the left-hand sides of two distinct rules in $R$. Formally,
this amounts to showing that $NOSU\!P(R) = \emptyset$. The main idea
of the proof is that we show that such superpositions would induce
critical pairs, and then these critical pairs cannot be joinable, which 
would contradict the confluence of our system ($R$ is convergent by definition). 

\begin{corollary}
\label{UnjoinableCriticalPairs}
$NOSU\!P(R) = \emptyset \,$ for all LM-systems~$R$.
\end{corollary}

\begin{proof}
We prove the following: 
if $R$ is an LM-System, $l_1 \to r_1$ and $l_2 \to r_2$ rules
in~$R$, $p \in \fpos(l_1^{})$ a non-root position, and $\sigma(
l_1^{} [l_2^{}]_p^{})$ a superposition, then the critical pair
$\langle \sigma ( l_1^{} [ r_2^{} ]_p^{} ), \; \sigma ( r_1^{} )
\rangle$ is not joinable modulo~$R$.

Assume towards deriving a contradiction that $s = \sigma ( l_1^{}
[ r_2^{} ]_p^{} )$ and $t = \sigma ( r_1^{} )$ are joinable
modulo $R$. Since $l_1^{} \to r_1^{} \in R$ there must be distinct
function symbols $f$ and~$g$ such that 
$s(\epsilon) = f$ and $t(\epsilon) = g$ as $R$ is an LM-System,
i.e., $l_1^{} \to r_1^{}$ must have root-pair $(f, g)$ for 
$f \not = g$.

Then, it must be that $s = f(s_1, \ldots, s_m)$ and $t = g(t_1,
\ldots, t_n)$. Since $s$ and $t$ are assumed to be joinable,
$l_1^{} \to r_1^{}$ is a rule in $R$ with root pair $(f, g)$ we
can apply Corollary~\ref{UniqueRule2} to get that:

\[ f(\widehat{s_1}, \ldots , \widehat{s_m}) ~ \mathop{\xrightarrow{\hspace*{0.75cm}}}_{l_1^{} \rightarrow r_1^{}}^{} ~ 
g(\widehat{t_1}, \ldots , \widehat{t_n}) \]

\noindent
where $\widehat{s_1}, \ldots, \widehat{s_m}, \, \widehat{t_1},
\ldots, \widehat{t_n}$ are the $R$-normal forms of $s_1, \ldots, s_m, \,
t_1, \ldots,
t_n$ resp. However, since $f(\widehat{s_1}, \ldots ,
\widehat{s_m})$ is $\bar{\epsilon}$-irreducible, the rule must be applied at
the root, which gives us that $f(\widehat{s_1}, \ldots ,
\widehat{s_m}) = \theta(l_1^{})$ for some substitution
$\theta$. Putting this all together we then get the following
reduction: \[ \sigma ( l_1^{} [ r_2^{} ]_p^{} ) \, \longrightarrow^{*}_{R}  \, \theta(l_1) \] Thus, 
Corollary~\ref{ODPReductions} applies and $\forall q \in
ODP(\sigma ( l_1^{} [ r_2^{} ]_p^{} ), \; \theta(l_1)) : \,
\left[ \vphantom{b_b^b} \sigma (
  l_1^{} [ r_2^{} ]_p^{} )|_q \rightarrow^{*}
  \theta(l_1)|_q \right]$. 
Since $p$ belongs to $ODP(\sigma ( l_1^{} [ r_2^{} ]_p^{} ), \; \theta(l_1))$,
it follows that  
$\sigma(r_2^{}) \rightarrow^{*} \theta(l_1)|_p = \theta(l_1|_p)$.
But since $l_1|_p$ is unifiable with $l_2$, we have
$l_1 |_p^{} (\epsilon) = l_1 (p) = l_2 (\epsilon)$
and this contradicts
Corollary~\ref{NoReversals}.
\end{proof}

We now turn to the case of forward-overlaps as
defined in the section on forward-closure. 

\begin{Lemma}
\label{NoForwardOverlaps}
Let $R$ be an LM-System, then $R \rightsquigarrow R = \varnothing$. 
\end{Lemma}

\begin{proof}
Recall: $R \rightsquigarrow R = \left \{(l_1 \rightarrow r_1)
\rightsquigarrow_p (l_2 \rightarrow r_2) \; \big| \; l_1 \rightarrow r_1, \,
l_2 \rightarrow r_2 \in R \; \wedge \; p \in \fpos(r_1) \vphantom{b^b} \right\}$ which are
not redundant in $R$. The proof proceeds by contradiction. Suppose
that the above set is non-empty. By
Corollary~\ref{noRootOverlaps}, no forward overlap can occur at $p =
\epsilon$.

Thus, there exists at least two rules, $l_1 \rightarrow r_1$ and $l_2
\rightarrow r_2$ in $R$, such that $p \in \fpos(r_1)$ is a non-root
position, and $\sigma ~ = ~ mgu(r_1|_p =^? l_2)$ exists. Suppose that
$l_2^{} \rightarrow r_2^{}$ has root-pair $(f, g)$. Forward
closure gives us the rule $\sigma(l_1^{}) \rightarrow \sigma(r_1[r_2]_p)$.
By~Corollary~\ref{UniqueRule2} the normal form of
$\sigma(l_1^{})$ and $\sigma(r_1^{})$
must be an instance of~$r_1^{}$, i.e., $\beta (r_1^{})$ for some
substitution~$\beta$. The normal form of $\sigma(r_1[r_2]_p)$
must also be the same~$\beta (r_1^{})$.
But note that $p \in ODP( \sigma(r_1[r_2]_p) ,
\, \beta (r_1^{}) )$ since $r_1^{} (p) = l_2^{} (\epsilon) = f \neq
r_2^{} (\epsilon) = g$. Thus 
$\sigma (r_2^{}) \, \rightarrow_{}^+ \, \beta (r_1^{})|_p^{}$. But, as mentioned,
the root symbol of~$r_1^{}|_p^{}$ (and hence that of~$\beta (r_1^{})|_p^{}$)
is the same as the root symbol of~$l_2^{}$. This contradicts
Corollary~\ref{NoReversals}.
\end{proof}
We can now state the following lemma, which follows easily from the above
results concerning overlaps. First we introduce the following definition: 
\begin{defn} A term-rewriting system $R$ is said to be \emph{non-overlapping} if and only if there are
no left-hand side superpositions and no forward-overlaps.
\end{defn}
\begin{Lemma}
Every LM-system is non-overlapping.
\end{Lemma}

\begin{proof}
Let $R$ be an LM-system. Suppose $l_i \rightarrow r_i \in R$ for $i \in \{1, 2\}$ are
distinct rules. Then, there can be no overlaps between $l_1$ and $l_2$ at position $p = \epsilon$ by
lemma~\ref{noRootOverlaps}. Lemma~\ref{UnjoinableCriticalPairs} establishes that $l_1$ and $l_2$ cannot
overlap at position $p \not = \epsilon$. Finally, by Lemma~\ref{NoForwardOverlaps}, no overlaps can occur between
$r_1$ and $l_2$.  
\end{proof}

Finally, using the results derived above about \emph{LM-systems}, we show
that every \emph{LM-system} is saturated by paramodulation. It is clear that 
every rewrite system saturated by paramodulation is also forward-closed.
The next result establishes that for \emph{LM-systems}, these two concepts are equivalent. 
Specifically, an \emph{LM-system} is trivially saturated by paramodulation as there can be no 
overlaps into the left-hand side of an equation nor the right-hand side of an equation. 

\begin{corollary}
\label{SaturatedByParamod}
Every LM-System is saturated by paramodulation.
\end{corollary}

\begin{proof}
Let $R$ be an LM-System and $E$ be the set of equations obtained
from~$R$. Suppose $u[s']_p \approx v$ in $E$ and $s ~ \rightarrow
~ t$ in $R$ induces a paramodulation inference. There are two
cases depending on whether $u[s']_p$ is the lhs or rhs of some
rule in~$R$. If it is the lhs, then this would contradict
Corollary~\ref{UnjoinableCriticalPairs}. If $u[s']_p$ is the rhs
of some rule in $R$, then this would contradict
Lemma~\ref{NoForwardOverlaps}.

Thus, each case leads to a contradiction, and so no
paramodulation inference steps can be performed.
\end{proof}

\section{The Cap Problem Modulo LM-Systems}

In this section we prove that although \emph{LM-systems} are a
restrictive subclass of term-rewriting systems there are still
important problems that are undecidable when restricted to
\emph{LM-systems}. Specifically, we show that the \emph{cap problem}\footnote{Also
known as the \emph{deduction problem}}, which
has important applications in cryptographic protocol analysis, is
undecidable even when the rewrite system~$R$ is an \emph{LM-system}.

The cap problem is defined as follows: 
\noindent \\

\underline{\underline{\textbf{Instance:}}} ~ A LM-System $R$, a set
$S$ of ground terms representing the intruder knowledge, and a ground
term~$M$.

\underline{\underline{\textbf{Question:}}} Does there exist a cap
term $C(\diamond_1, \ldots, \diamond_n)$ such that $C[\diamond_1:=
  s_{i_1}, \ldots, \diamond_n := s_{i_n}] \, \rightarrow^{*}_{R} \,
M$? \\

We show that the above problem is undecidable by a many-one
reduction from the halting problem for reversible deterministic
2-counter Minsky machines (which are known to be equivalent to
Turing machines). The construction below is extremely similar to
the one given in~\cite{NotesOnBSM}. Originally, the construction
was used to show the undecidablility of the subterm-collapse
problem for LM-Systems. Here it is modified slightly to account
for the cap problem, however the majority of the rules remain
unchanged.

A reversible deterministic 2-counter Minsky machine (henceforth a
Minsky machine) is described as a tuple $N = (Q, \delta, q_0, q_L)$
where $Q$ is a finite non-empty set of states, $q_0, q_L \in Q$ are
the initial and final states respectively and $\delta$ is the transition
relation. The elements of the transition relation $\delta$ are
represented as 4-tuples of the following 
form: \[ \left[q_i, j, k, q_{i'} \right] \text{  or  } \left[q_i, j, d, q_{i'} \right] \]

\noindent
where $q_i, q_i' \in Q, \; j \in \{1, 2\}, \; k \in \{Z, P\}, \;
d \in \{0, +, -\}$.  Tuples of the first form represent that the
machine is in state $q_i$, checks if counter $j$ is zero $(Z)$ or
positive $(P)$ and transitions to state $q_{i'}$. Tuples of the
second form represent that the machine is in state $q_i$, and
either decrements $(-)$, increments $(+)$, or does nothing $(0)$
to counter $j$ and transitions to state~$q_{i'}$.

Each configuration of machine $N$ is written as $(q_i, C_1, C_2)$
where $q_i \in Q$ is the current state of the machine, and $C_1,
C_2$ are the values of the counters.  We encode such
configurations as terms. The initial and final configurations of
$N$ are encoded by: $c(q_0, s^k(0), s^p(0), 0)$ and 
$c(q_L, s^{k'}(0), s^{p'}(0), s^n(0))$ respectively. The fourth
argument of~$c$ corresponds to the number of steps the machine has
taken.

We need the fact that $N$ is deterministic and reversible in the sequel to 
establish the results that the construction provided actually produces
an LM-System. Namely, we need the following fact: For every pair of
tuples in $\delta$, $[q_{i_1}, j_1, x_1, q_{i'_1}]$ and $[q_{i_2}, j_2, 
x_2, q_{i'_2}]$ we have that 

\[ (i_1 = i_2) \vee (i'_1 = i'_2) \Rightarrow (j_1 = j_2 \wedge \{x_1, x_2\} = \{Z, P\}) \]

This means that $N$ can leave or enter the same state on two different transitions only 
when the same counter is being checked and different checks are being performed on that
counter.

Let $N$ be a Minsky machine. We construct a term-rewriting system
$R_N$ over the signature \[ \Sigma \, = \, \bigcup\limits_{i = 1}^{L} \{f_i,
f'_i, q_i \} \cup \left\{c, s, 0, g, g', e \right\} \] and show that the
resulting TRS is an LM-System such that if $N$ starts in the initial
configuration and halts in a final configuration then there exists a
cap term $C$ such that the ground term $c(e, 0, 0, 0)$ (playing the
role of $M$ in the description of the problem above) can be
deduced. We begin the construction by initializing $R_N$ with the
following rules:
\begin{align*}
	f_L(c(q_L, s^{k'}(0), s^{p'}(0), z)) \rightarrow g(c(e, 0, 0, z)) \\[+4pt]
	g'(g(c(e, 0, 0, s(z)) \rightarrow c(e, 0, 0, z)
\end{align*}

We then add the following rules to $R_N$, 
each of which encodes a possible move of the machine. 
That is, each rule represents an element of $\delta$. 

\begin{enumerate}
    \item[(a1)] $\, [ \, q_i, 1, P, q_j \, ]\colon R_N \, := \,
        R_N \cup \{ \, f_i(c(q_i, s(x), y, z)) \to c(q_j, s(x), y, s(z)) \, \}$

    \item[(a2)] $\, [ \, q_i, 2, P, q_j \, ]\colon R_N \, := \,
        R_N \cup \{ \, f_i(c(q_i, x, s(y), z)) \to c(q_j, x, s(y), s(z)) \, \}$
        \vspace{2ex}

    \item[(b1)] $\, [ \, q_i, 1, Z, q_j \, ]\colon R_N \, := \,
        R_N \cup \{ \, f'_i(c(q_i, 0, y, z)) \to c(q_j, 0, y, s(z)) \, \}$

    \item[(b2)] $\, [ \, q_i, 2, Z, q_j \, ]\colon R_N \, := \,
        R_N \cup \{ \, f'_i(c(q_i, x, 0, z)) \to c(q_j, x, 0, s(z)) \, \}$
        \vspace{2ex}

    \item[(c1)] $\, [ \, q_i, 1, +, q_j \, ]\colon R_N \, := \,
        R_N \cup \{ \, f_i(c(q_i, x, y, z)) \to c(q_j, s(x), y, s(z)) \, \}$

    \item[(c2)] $\, [ \, q_i, 2, +, q_j \, ]\colon R_N \, := \,
        R_N \cup \{ \, f_i(c(q_i, x, y, z)) \to c(q_j, x, s(y), s(z)) \, \}$
        \vspace{2ex}

    \item[(d1)] $\, [ \, q_i, 1, 0, q_j \, ]\colon R_N \, := \,
        R_N \cup \{ \, f_i(c(q_i, x, y, z)) \to c(q_j, x, y, s(z)) \, \}$

    \item[(d2)] $\, [ \, q_i, 2, 0, q_j \, ]\colon R_N \, := \,
        R_N \cup \{ \, f_i(c(q_i, x, y, z)) \to c(q_j, x, y, s(z)) \, \}$
        \vspace{2ex}

    \item[(e1)] $\, [ \, q_i, 1, -, q_j \, ]\colon R_N \, := \,
        R_N \cup \{ \, f_i(c(q_i, s(x), y, z)) \to c(q_j, x, y, s(z)) \, \}$

    \item[(e2)] $\, [ \, q_i, 2, -, q_j \, ]\colon R_N \, := \,
        R_N \cup \{ \, f_i(c(q_i, x, s(y), z)) \to c(q_j, x, y, s(z)) \, \}$
\end{enumerate}
\vspace{1ex}

We now state the following theorem, and hold off on providing a proof until 
various claims have been shown to hold. The result will then following as an easy corollary.

\begin{theorem}
	\label{CapUndec}
	The cap problem modulo LM-Systems is undecidable. 
\end{theorem}

We first begin by showing that given a reversible deterministic 2-counter Minsky machine, $N$,
the rewrite system constructed above, $R_N$, is convergent.  

\begin{claim}
\label{RNConvergence}
Let $N$ be a Minsky machine, then the TRS, $R_N^{}$, is convergent. 
\end{claim}

\begin{proof}
Define $\succ$ on $\Sigma$ as follows: $f_i \succ f'_i \succ f_L \succ
g \succ g' \succ c \succ s \succ q_j \succ e \succ 0$.  Then
termination of $R_N$ is apparent by applying a recursive path ordering
induced by $\succ$. We show that $R_N$ is confluent by showing that
$R_N$ has no critical pairs. By construction, there are clearly no
superpositions that can occur at positions other than than
$\epsilon$. However, by the determinism of $N$, if the index $i$
occurs more than once as the subscript of the lhs of a rule with an
$f$ as its root symbol, then it could only ever occur again as the
index of a term with an $f'_{i}$ as the root symbol of the lhs. Thus,
there can be no critical pairs.
\end{proof}

The next claims establish that $R_N$ is also forward closed 
and that $RHS(R_N)$ is quasi-deterministic. Each of which is 
a condition of a rewriting system to be an LM-System. 

\begin{claim}
\label{RNFowardClosed}
Let $N$ be a Minsky machine, then $R_N$ is forward-closed. 
\end{claim}

\begin{proof}
No root-symbol of the lhs of any rule in $R_N$ occurs in the rhs of any other rule. 
That is, there is no way to unify a subterm of the rhs of any rule with the lhs of any other rule. 
Thus, there can be no forward-overlaps and therefore $R_N$ is forward closed. 
\end{proof}

\begin{claim} 
\label{RNQDet}
Let $N$ be a Minsky machine, then $RHS(R_N)$ is quasi-deterministic. 
\end{claim}

\begin{proof}
We first show that there are no RHS overlaps in $R_N$. Suppose $l_i
\rightarrow r_i \in R_N$ for $i \in \{1, 2\}$ induces a RHS overlap.
By construction of $R_N$, it must be that $r_1(\epsilon) = c =
r_2(\epsilon)$ and $r_1(1) = q_i = r_2(1)$. Since $N$ is
deterministic, the only way this could occur is between a rule from
set $(a)$ and a rule from set $(b)$, but then $r_1$ and $r_2$ are not
unifiable as there would be a function clash between ``$s$" and
``0". Thus, there are no RHS overlaps.

It then suffices to show that $R_N$ itself is
quasi-deterministic. Clearly, no rule contains a variable as its
left-hand side or its right-hand side, and no rule is root-stable. It
remains to show that there are no root-pair repetitions. Suppose $l_i
\rightarrow r_i \in R_N$ for $i \in \{1, 2 \}$ induces a root-pair
repetition. Suppose $(h, c)$ is the repeated root-pair.  This implies
that there are some pair of 4-tuples $x, y \in \delta$ such that the
first coordinates of $x$ and $y$ where the same. Let $x = [q_\alpha, j_1, m,
  q_p]$ and $y = [q_\alpha, j_2, n, q_k]$ be such tuples.  Since $N$ is
deterministic and reversible, it must be that $j_1 = j_2$ and $\{m,
n\} = \{Z, P\}$, but then $l_1 \rightarrow r_1$ would have root-pair
$(f_\alpha, c)$ and $l_2 \rightarrow r_2$ would have root pair~$(f'_\alpha, c)$.

Thus, assuming there is a root-pair repetition contradictions the
definition of the construction of $R_N$ from $\delta$. Therefore,
there are no root-pair repetitions in $R_N$, and therefore $RHS(R_N)$
is quasi-deterministic.
\end{proof}

Finally, the claim below, along with the claims above, establish that
$R_N$ is an LM-System. Namely, we prove that $R_N$ is non-subterm 
collapsing, and thus we can conclude that $RHS(R_N)$ is deterministic. 
Then, we establish that $N$ halts if and only if there exists a cap-term 
that allows the deduction of $M$ modulo $R_N$. Thus, putting it all
together we establish a many-one reduction and provide a proof of 
Theorem~\ref{CapUndec}. 

\begin{claim}
\label{RNNonSubterm}
Let $N$ be a Minsky machine, then $R_N$ is non-subterm-collapsing. 
\end{claim}

\begin{proof}
Suppose $t \rightarrow_{R_N}^{+} t'$ such that $t'$ is a proper
subterm of~$t$.  By construction of $R_N$,~ $t|_p(\epsilon) = c$ for
some position~$p$. Since rules $(a1) - (e2)$ and the second rule in
the initialization of $R_N$ produce a ``$c$" term with the last argument
with an extra $s$ added or removed, the only rule that could have been
used to produce the collapse was the first rule of the initialization
of~$R_N$.  But in this case, the $q_L$ is replaced with an $e$ and
thus the resulting term in the reduction could not be a proper subterm
of~$t$. Thus we have a contradiction and can conclude that $R_N$ is
non-subterm-collapsing.
\end{proof}

\begin{claim}
\label{RNLMSystem}
 Let $N$ be a Minsky machine. Then the TRS $R_N^{}$ is an LM-System.
\end{claim}

\begin{proof}
Claim~\ref{RNConvergence} shows that $R_N$ is convergent and
claim~\ref{RNFowardClosed} shows that $R_N$ is
forward-closed. Claims~\ref{RNQDet} and \ref{RNNonSubterm} show that
$RHS(R_N)$ is deterministic. Therefore, $R_N$ is a LM-System.
\end{proof}

\begin{claim}
\label{CapiffHalt}
Let $N$ be a Minsky machine. Then starting in the initial configuration
\( \left( \vphantom{b_b^b} q_0^{}, k, p \right) \)
$N$ halts in
\( \left( \vphantom{b_b^b} q_L^{}, k', p' \right) \)
iff there exists a cap term $C(\diamond)$ such that 
if $M = c(e, 0, 0, 0)$ and $S = \{c(q_0, s^k(0), s^p(0), 0) \}$ 
then $C[ \diamond := c(q_0, s^k(0), s^p(0), 0)] \, \rightarrow_{R_N}^{*} \, M$. 
\end{claim}

\begin{proof}
$(\Rightarrow)$ Suppose $N$ halts, then it does so in a finite number
  of steps. Let the number of steps $N$ takes to be $n \in
  \mathbb{N}$. The proof then proceeds similiarly to that of Lemma 4.4
  in~\cite{NotesOnBSM}. That is, let $(\tau_i)_{i = 1}^{n}$ be the
  sequence of transitions that $N$ passes through on its computation
  run. Then, $\tau_1 = [q_0, j, d, q_i]$ and $\tau_n = [q_{i'}, j, d,
    q_L]$. Define $f^{*}_{j}$ as follows:

$$
f^*_i = 
\begin{cases}
f_{j}^{'} & \mbox{ if } d = Z \\
f_j & \mbox{otherwise}
\end{cases}
$$

We can then construct the cap term $C(\diamond)$ as follows: 

\[ C(\diamond) = (g' \circ g)^{n-1} ( g'(f_L((f^*_n \circ \dots \circ f^*_1)(\diamond)))) \]

Since the rules of $R_N$ simulate the sequence of configurations of $N$ and $N$ halts, we can set 
$C[\diamond := c(q_0, s^k(0), s^p(0), 0)]$ and use the rules of $R_N$ to reduce this term to \[ C(\diamond) \rightarrow_{R_N}^{+} (g' \circ g)^{n-1} (g'(g(c(e, 0, 0, s^n(0))) \] which can then be reduced using the second rule of $R_N$ to get 
the term $(g' \circ g)^{n-1}(c(e, 0, 0, s^{n-1}(0)))$. This can then 
be reduced by further applications of the same rule to get the term 
$g'(g(c(e, 0, 0, s(0))))$. At this point, one more application of the second
rule will yield $c(e, 0, 0, 0) = M$. Thus, $M$ can be deduced. 

$(\Leftarrow)$ Suppose that there exists a cap $C(\diamond)$ such that
$C[\diamond := c(q_0, s^k(0), s^p(0), 0)] \rightarrow_{R_N}^{*} M$.
Let $t$ be such a cap term. Thus, the above says that there is a
reduction chain starting with $t$ and ending in the term $M$. That is,
$t \rightarrow t_1 \rightarrow t_2 \rightarrow \dots \rightarrow
t_{\kappa} = c(e, 0, 0, 0)$.  Since the only rule that can introduce
an ``$e$" term is the first rule, there must be a sub-chain of
reductions $t \rightarrow t_1 \rightarrow \dots \rightarrow t_i$ such
that $t_i = \sigma(f_L(c(q_L, s^{k'}(0), s^{p'}(0), z))$.

By design of the rewrite system $R_N$ from $N$, one can see that the
chain above corresponds to a sequence of configurations of $N$ that
eventually leads to the configuration $(q_L, k', p')$.  Therefore,
$N$, when started in $(q_0, k, p)$, halts in configuration~$(q_L, k',
p')$.
\end{proof}
We can now state the proof of Theorem~\ref{CapUndec}
\begin{proof}
Given a Minsky machine $N$, let $G$ be the function such that $G(N) \rightarrow R_N$, i.e. the function that
takes as input a Minsky machine and produces the corresponding system $R_N$ as described above. Then $G$ is clearly
a total recursive function. Claim~\ref{CapiffHalt} says then that $G$ is in 
fact a many-one reduction. Thus, since the halting problem for Minsky 
machines is undecidable, then so is the cap problem modulo LM-Systems. 
\end{proof}

\ignore{

}

\bibliographystyle{plain}
\bibliography{ref2}

\begin{thebibliography}{1}

\bibitem{Term}
Franz Baader and Tobias Nipkow.
\newblock {\em Term rewriting and all that}.
\newblock Cambridge university press, 1999.

\bibitem{BaaderSnyd-01}
Franz Baader and Wayne Snyder.
\newblock Unification theory.
\newblock {\em Handbook of automated reasoning}, 1:445--532, 2001.

\bibitem{BGLN}
Christopher Bouchard, Kimberly~A. Gero, Christopher Lynch, and Paliath
  Narendran.
\newblock On forward closure and the finite variant property.
\newblock In Pascal Fontaine, Christophe Ringeissen, and Renate~A. Schmidt,
  editors, {\em FroCos}, volume 8152 of {\em Lecture Notes in Computer
  Science}, pages 327--342. Springer, 2013.

\bibitem{BHSS}
Hans-J{\"u}rgen B{\"u}rckert, Alexander Herold, and Manfred Schmidt-Schauss.
\newblock On equational theories, unification, and (un) decidability.
\newblock {\em Journal of Symbolic Computation}, 8(1-2):3--49, 1989.

\bibitem{NotesOnBSM}
Kimberly Gero, Chris Bouchard, and Paliath Narendran.
\newblock Some notes on basic syntactic mutation.
\newblock In Santiago Escobar, Konstantin Korovin, and Vladimir Rybakov,
  editors, {\em UNIF 2012 Post-Worskhop Proceedings. The 26th International
  Workshop on Unification}, volume~24 of {\em EPiC Series in Computing}, pages
  17--27. EasyChair, 2014.

\bibitem{gramlich}
Bernhard Gramlich.
\newblock On interreduction of semi-complete term rewriting systems.
\newblock {\em Theoretical Computer Science}, 258(1):435--451, 2001.

\bibitem{Hermann}
Miki Hermann.
\newblock Chain properties of rule closures.
\newblock {\em Formal Aspects of Computing}, 2(1):207--225, 1990.

\bibitem{LynchMorawska}
Christopher Lynch and Barbara Morawska.
\newblock Basic syntactic mutation.
\newblock In {\em Automated Deduction (CADE-18)}, pages 471--485. Springer,
  2002.

\bibitem{Ohlebusch95}
Enno Ohlebusch.
\newblock Termination is not modular for confluent variable-preserving term
  rewriting systems.
\newblock {\em Information processing letters}, 53(4):223--228, 1995.

\end{thebibliography}


\begin{thebibliography}{}

\bibitem{AbadiC06}
M.~Abadi, V.~Cortier.
\newblock Deciding knowledge in security protocols under equational theories.
\newblock \emph{Theoretical Computer Science}~367(1-2):2--32, 2006.


\bibitem{SerdarEtAl12}
S.~Anantharaman, S.~Erbatur, C.~Lynch, P.~Narendran, M.~Rusinowitch.
\newblock ``Unification modulo Synchronous Distributivity''. 
\newblock Technical Report~SUNYA-CS-12-01, Dept. of Computer Science,
University at Albany---SUNY, 2012. Available at
{\verb!www.cs.albany.edu/~ncstrl/treports/Data/README.html!} 
(An abridged version to be presented at {\em IJCAR~2012.\/})

\bibitem{Term} F.~Baader, T. Nipkow.
\newblock {\em Term Rewriting and All That}.
\newblock Cambridge Univ Press, 1999.

\bibitem{BaaderSnyd-01} F.~Baader, W.~Snyder.
\newblock ``Unification Theory''.
\newblock {\em Handbook of Automated Reasoning\/}, pp. 440--526, 
Elsevier Science Publishers B.V., 2001. 

\bibitem{BGLN}
C.~Bouchard, K.\space A.~Gero, C.~Lynch, P.~Narendran.
\newblock On forward closure and the finite variant property.
In Pascal Fontaine, Christophe Ringeissen, and Renate~A. Schmidt,
  editors, \emph{FroCoS}, volume~8152 of \emph{Lecture Notes in Computer
  Science}, pages 327--342. Springer, 2013.

\bibitem{NotesOnBSM}
C.~Bouchard, K.\space A.~Gero, and P.~Narendran.
\newblock Some Notes on Basic Syntactic Mutation.
\newblock In S.~Escobar, K.~Korovin, and V.~Rybakov, editors, \emph{Proceedings
  26th International Workshop on Unification,} pages 9--14, 2012.




\bibitem{BHSS} H-J.\ B\"{u}rckert, A.\ Herold, M.\
Schmidt-Schau{\ss}.  
\newblock On Equational Theories, Unification,
and (Un)Decidability.  
\newblock {\em Journal of Symbolic
Computation\/}~8(1/2): 3-49 (1989).

\bibitem{gramlich} B.\ Gramlich.
\newblock On interreduction of semi-complete term rewriting systems. 
\newblock \emph{Theoretical Computer Science}~258(1-2): 435-451~(2001).

\bibitem{GKM83} J.V.\ Guttag, D.\ Kapur, D.R.\ Musser.
\newblock On Proving Uniform
Termination and Restricted Termination of Rewriting Systems. 
\newblock {\em SIAM J. Computing\/}~12(1): 189--214 (1983).

\bibitem{Hermann} M.\ Hermann.
\newblock Chain properties of rule closures.
\newblock {\em Formal Aspects of Computing\/}~2: 207--225 (1990).

\ignore{
\bibitem{HKMV} 
G.G.\ Hillebrand, P.C.\ Kanellakis, H.G.\ Mairson, M.Y.\ Vardi.
\newblock Undecidable Boundedness Problems for Datalog Programs. 
\newblock {\em J. Logic Programming\/}~25(2): 163-190 (1995).
}

\bibitem{LynchMorawska} C.\ Lynch, B.\ Morawska.
\newblock Basic Syntactic Mutation.
\newblock In A.~Voronkov, editor, \emph{CADE}, Volume~2392 of \emph{Lecture Notes
  in Computer Science,} pages 471--485. Springer, 2002.

\bibitem{Metivier} Y.\ Métivier.
  \newblock About the Rewriting Systems Produced by the Knuth-Bendix Completion Algorithm.
  \newblock \emph{Information Processing Letters}~ 16(1): 31-34~(1983).

\bibitem{Ohlebusch95}
E.~Ohlebusch.
\newblock {Termination is not Modular for Confluent Variable-Preserving Term
  Rewriting Systems}.
\newblock {\em Information Processing Letters\/}~53(4):223--228, 1995.

\end{thebibliography}

\section{Appendix}

\begin{Lemma}
It is not the case that the equational theory of every $LM$-system
is {\em finite\/} (i.e., all congruence classes are finite).
\label{inf-cong}
\end{Lemma}

\begin{proof}
The system \[ f(g(h(x))) \rightarrow g(x) \] is an $LM$-system, but the
congruence class of $g(y)$ is clearly infinite.
\end{proof}

\begin{Lemma}
Every $LM$-system is free over the signature.
\label{free-sig}
\end{Lemma}

\begin{proof}
Let $R$ be an $LM$-system and $f \in \Sigma^{(n)}$.
Let $s~=~f( s_1^{} , \ldots , s_n^{} )$ and $t = f( t_1^{} , \ldots , t_n^{} )$
be $\bar{\epsilon}$-irreducible terms that are joinable modulo~$R$. (Thus
$s_1^{} , \ldots , s_n^{}, t_1^{} , \ldots , t_n^{}$ are in normal form.)
Then either $t$ is in normal form and $s \to t$, or $s \to \widehat{t}$ and $t \to \widehat{t}$
where $\widehat{t}$ is the normal form of $s$ and $t$. (Note that since
$R$ is forward-closed, reduction to a normal form will take only one step.)
But the former is impossible because no rule in~$R$ can have
the same root symbol on both the left-hand side
and the right-hand side. The latter is ruled out because
no two rules can have the same root symbols at their
sides, i.e., root-pair repetitions are not allowed.
\end{proof}

\begin{Lemma}
\label{UniqueEquation}
Let $R$ be an $LM$-system and $s = f(s_1^{}, \ldots , s_m^{})$ and $t
= g(t_1^{}, \ldots , t_n^{})$ be $\bar{\epsilon}$-irreducible terms such
that~$f \neq g$.  Then $s$ and $t$ are joinable modulo~$R$ if and only
if there is a unique \emph{equation} $e_1^{} \approx e_2^{} \; \in \;
RHS(R)$ with root pair $(f, g)$ such that \( \displaystyle{s ~
  \mathop{\xrightarrow{\hspace*{0.75cm}}}_{e_1^{} \rightarrow e_2^{}}^{} ~ t}. \) 
\end{Lemma}

\begin{proof}
The result follows from Lemma~\ref{ExactlyTwo} and its proof.
\end{proof}


\end{document}